\documentclass[]{article}

\usepackage[english]{babel}
\usepackage[utf8]{inputenc}

\usepackage{amssymb}                
\usepackage{amsmath}                
\usepackage{amsfonts}
\usepackage{amsthm}

\usepackage{graphicx}
\usepackage[colorinlistoftodos]{todonotes}
\usepackage{hyperref}

\newtheorem*{mydef}{Definition}
\newtheorem*{thm}{Theorem}
\newtheorem*{cor}{Corollary}

\usepackage{authblk} 
\usepackage{appendix} 
\usepackage{fullpage} 

\usepackage{bibunits} 
\defaultbibliography{lin_dg.bib} 
\defaultbibliographystyle{unsrt} 

\usepackage{chngcntr} 
\counterwithin{equation}{section}

\usepackage{todonotes} 

\title{Cancer treatment scheduling and dynamic heterogeneity in social dilemmas of tumour acidity and vasculature}

\author[1]{Artem Kaznatcheev}
\author[2]{Robert Vander~Velde}
\author[1,3]{Jacob G. Scott}
\author[1]{David Basanta}

\affil[1]{Integrated Mathematical Oncology, H. Lee Moffitt Cancer Center and Research Institute, Tampa, FL, USA}
\affil[2]{Department of Molecular Medicine, Morsani College of Medicine, University of South Florida, Tampa, FL, USA}
\affil[3]{Wolfson Centre for Mathematical Biology, Mathematical Institute, University of Oxford, Oxford, UK}

\date{\today}

\newcommand{\GLY}[0]{{\it GLY}}
\newcommand{\VOP}[0]{{\it VOP}}
\newcommand{\DEF}[0]{{\it DEF}}

\begin{document}
\begin{bibunit} 
\maketitle

\begin{abstract}
{\bf Background:}
Tumours are diverse ecosystems with persistent heterogeneity in various cancer hallmarks like self-sufficiency of growth factor production for angiogenesis and reprogramming of energy-metabolism for aerobic glycolysis.
This heterogeneity has consequences for diagnosis, treatment, and disease progression.

{\bf Methods:}
We introduce the double goods game to study the dynamics of these traits using evolutionary game theory.  
We model glycolytic acid production as a public good for all tumour cells and oxygen from vascularization via VEGF production as a club good benefiting non-glycolytic tumour cells. 
This results in three viable phenotypic strategies: glycolytic, angiogenic, and aerobic non-angiogenic.

{\bf Results:}
We classify the dynamics into three qualitatively distinct regimes: (1) fully glycolytic, (2) fully angiogenic, or (3) polyclonal in all three cell types. 
The third regime allows for dynamic heterogeneity even with linear goods, something that was not possible in prior public good models that considered glycolysis or growth-factor production in isolation.

{\bf Conclusion:}
The cyclic dynamics of the polyclonal regime stress the importance of timing for anti-glycolysis treatments like lonidamine.
The existence of qualitatively different dynamic regimes highlights the order effects of treatments.
In particular, we consider the potential of vascular renormalization as a neoadjuvant therapy before follow up with interventions like buffer therapy.

\end{abstract}

\newpage

\section{Introduction}

Tumours are highly heterogeneous ecosystems~\cite{HetTheory}, with various cancerous and non-cancerous sub-populations of cells competing for access to space, growth-factors, nutrients, oxygen, and other limited resources.
This existence and persistence of heterogeneity has implications for diagnosis, treatment, and disease progression.\cite{DrugHet,ProgHet,TreatHet}
Yet, if tumour evolutionary dynamics proceed via clonal selection~\cite{cancerEvo1,cancerEvo2,cancerEvo3} then how can more than one clone stably coexist?
Current explanations of intra-tumour heterogeneity include evolutionary neutrality~\cite{neuralHet}, niche specialization~\cite{nicheHet1,nicheHet2}, non-equilibrium dynamics~\cite{nonEqHet}, and frequency dependent selection~\cite{EGTOnco1,EGTOnco2,glyEGT1}.  
It remains an open problem to identify which, or even how many, of these mechanisms are at work in any given neoplasm ~\cite{cancerEvo3}.

The progression of neoplasms to metastatic disease is marked by the acquisition of a number of hallmarks~\cite{hm_old,hm_new}, including self-sufficiency of growth factor production for angiogenesis and reprogramming energy-metabolism for aerobic glycolysis.
As with many of the other hallmarks, there is evidence of intra-tumour heterogeneity in both the production of cytokines like vascular endothelial growth factor (VEGF)~\cite{VEGF_Het,GenHet}, and glycolysis~\cite{glyHetBreast1,glyHetBreast2,glyHetLung}.
Given that it is possible for an individual cancer cell to not invest (as heavily) in angiogenesis or not forgo the benefit of oxygen by avoiding aerobic glycolysis then how do these population level traits evolve, and how are they maintained?
We answer this question with a mathematical model that treats acid production through glycolosysis as a tumour-wide public good that is coupled to the club good of oxygen from better vascularisation.

By investing in better vascularisation -- by over-producing VEGF, for example -- the whole tumour can benefit from an increase (or improvement) in vascularisation and the associated rise in availability of nutrients and oxygen~\cite{angiogensis1}.
An individual cancer cell, however, could reap these benefits from mere proximity to (over)-producers and save on the energetic cost of producing the relevant growth-factor: free-riding on the benefits created by other cancer cells.
These free-riding cancer cells could out-compete the (over)-producers and take over the tumour.
Such a switch away from growth factor (over)-production, however, would decreases the overall fitness of the tumour by making fewer nutrients and oxygen available to all cancer cells. 
What is favourable to individual cancer cell is disfavourable to the society of cancer cells that make up the tumour.
This represents a classic example of an evolutionary social dilemma.

A similar social dilemma exists for the increase of acidification from glycolysis.
It is striking that the up-regulation of glycolysis — the so called Warburg effect~\cite{Warburg1,Warburg2,Warburg3} — occurs even in the absence of hypoxia~\cite{aeroGly1,aeroGly2,hm_new}. 
Glycolysis is inefficient when oxygen is not a limiting factor, raising the conundrum of what selective advantage it provides to a cell to compensate for its energetic cost. 
The acid mediated tumour invasion hypothesis suggests that this advantage comes from the acidification of the tumour microenvironment that leads to an increase in normal cell death and higher invasiveness through increased matrix  degradation~\cite{amti1,amti2,amti3,amti4}. 
In order to benefit from this microenvironmental acidification, the cancer cells need to be resistant to microenvironmental acidity.
Early models have assumed that this resistance is available only to glycolytic cells,~\cite{amti2,glyEGT1,glyEGT2} but there is little evidence to suggest that aerobic cancer cells could not also develop this resistance.
In the absence of hypoxia, such resistant non-glycolytic cells could benefit from aerobic respiration to out-compete the glycolytic cancer cells. 
But such a switch away from glycolysis decreases the overall acid production by the tumour, negating part of the advantage acid-resistant cells have against non-cancerous soma.

Evolutionary game theory (EGT) is a tool available to theoretical biologists to make sense of these sort of social dilemmas. 
Originating with Maynard Smith \& Price~\cite{firstEGT}, EGT is a mathematical approach to modeling frequency-dependent selection where sub-populations interact through phenotypic strategies.
We continue the trend of the earliest applications of EGT to oncology~\cite{EGTOnco1,EGTOnco2} by focusing on the persistence of heterogeneity. 
In recent years, EGT has been used more extensively in oncology to study the conditions that select for more aggressive tumour phenotypes in gliomas~\cite{glyEGT1,glyEGT2}, colorectal cancer~\cite{carcinogenesis1,carcinogenesis2}, multiple myeloma~\cite{MMEGT} and prostate cancer~\cite{prostate}; 
as well as the effects of treatment on the progression of cancer~\cite{treatment1,treatment2}.
Most of this prior work focuses on matrix games, with pairwise interactions between cell types.

Recently, Archetti introduced the multiplayer public goods game to oncology for looking separately at two-strategy problems like growth-factor production~\cite{A13} and the Warburg effect~\cite{A14}.
And with colleagues, he has implemented the growth-factor production game in an experimental system~\cite{A15}.
He has stressed the importance of non-linear benefits from these goods for maintaining heterogeneity.
However, angiogenesis and glycolosysis are intimately related and should not be considered in separation because the benefits of oxygen affect the degree of hypoxia and thus the relative cost/benefit of glycolysis when compared to aerobic metabolism. 
We focus on this inter-dependence of these two hallmarks to continue this research programme by coupling the two goods — the public good of acidification and the club good of vascularization —  in a three-strategy game. 
Our model reveals dynamics that cannot be predicted from treating micro-environmental acidification and vascularization in isolation from each other. 
In particular, we show that heterogeneity is possible with linear goods and that this heterogeneity is not a static equilibrium but a dynamic cycle of constantly changing proportions of cancer cell types. 
This has consequences for the design of treatments, as it suggests that the timing, and order, of therapeutic interventions could drastically affect the outcome.
For example, it shows the importance of preparatory treatment or neoadjuvant therapy that manages the low-frequency cell types in contrast to simply targeting the most common clone, and of the advantages of targeting the tumour microenvironment instead of just targeting the cancer cell. 

\section{Double Goods Game}
\label{sec:game}

Acidity provides a (relative) benefit to all acid-resistant tumour cells that are competing against the acid-sensitive non-cancerous cells~\cite{ba1,ba2,acidSense}, regardless of whether they have aerobic or anaerobic metabolism.
Consider a focal glycolytic cell -- a cell undergoing anaerobic metabolism, even in an aerobic environment -- interacting with $n$ other nearby cells of which $n_G \in [0,n]$ are also glycolytic. 
Together they produce a relative benefit $b_a(n_G + 1)$ due to acidity to be distributed among the $n + 1$ cells. Therefore, this focal glycolytic cells receives a net benefit of $\frac{b_a(n_G + 1)}{n + 1}$.

By averaging over all possible focal glycolytic cells and interaction group compositions, we get that the expected fitness of a population of glycolytic (\GLY) cells with random assortment is:

\begin{equation}
w_G  = \langle \frac{b_a (n_G + 1)}{n + 1} \rangle_{n_G \sim \mathbf{B}_n(x_G)} \label{eq:wG}
\end{equation}

\noindent where $x_G$ is the proportion of \GLY\; in the population, and the angle brackets represent averaging with, in this case, $n_G$ sampled from the binomial distribution with $n$ trials and $x_G$ as the probability of success (i.e. choosing a \GLY\; cell).

For an aerobic cell -- non-\GLY\; -- in a similar case, the benefit due to acid is only $ \frac{b_a n_G}{n + 1}$ since it does not itself produce acid.
But, in addition to the effects of acid, an aerobic cell can also benefit from oxygen delivered by the vasculature.
In particular, a VEGF (over)-producer (\VOP) will receive the benefit $ \frac{b_v (n_V + 1)}{n - n_G + 1}$ but pay a cost $c$ for the higher production, and the aerobic non-(over)-producer of VEGF (\DEF) will receive the benefit $\frac{b_v n_V}{n - n_G + 1}$ but pay no cost, essentially free-riding.
Note that, unlike acidity, the benefit of oxygen from vasculature is divided among only the non-glycolytic cells in the interacting group, thus we are treating it as a club good~\cite{clubgood}.
Alternatively, this makes the glycolytic cells similar to loners in the optional public-goods game~\cite{VPG}.

Therefore, the expected fitnesses of the two aerobic populations with random assortment are:

\begin{align}
w_V & = \langle \frac{b_a n_G}{n + 1} \rangle_{n_G \sim \mathbf{B}_n(x_G)} 
	+ \langle \frac{b_v (n_V + 1)}{n - n_G + 1} \rangle_{ n_G,n_V \sim \mathbf{M}_n(x_G,x_V)} - c \label{eq:wV}\\
w_D & = \langle \frac{b_a n_G}{n + 1} \rangle_{n_G \sim \mathbf{B}_n(x_G)}
	+ \langle \frac{b_v n_V}{n - n_G + 1} \rangle_{n_G,n_V \sim \mathbf{M}_n(x_G,x_V)} \label{eq:wD}
\end{align}

\noindent where $x_V, x_D$ are the proportions of \VOP\;and \DEF\;in the population, and the averages in the second summands are taken with $n_G, n_V$ sampled from the multinomial distribution with $n$ trials and $x_G$ as the probability of the first outcome (i.e. choosing a \GLY\; cell), and $x_V$ as the probability of the second outcome (i.e. choosing a \VOP\; cell).
In each equation, the first summand is the benefit due to acidification and the second is the benefit from the club good of vascularization.
These fitness functions are described in more detail in appendix~\ref{app:genModel}.

The evolutionary dynamics of the population are given by the replicator equation~\cite{repDyn1,repDyn2}:

\begin{align}
\dot{x}_G & = x_G(w_G - \langle w \rangle) \label{eq:xGgen}\\
\dot{x}_V & = x_V(w_V - \langle w \rangle) \label{eq:xVgen}\\
\dot{x}_D & = x_D(w_D - \langle w \rangle) \label{eq:xDgen}
\end{align}

\noindent where $\langle w \rangle = x_G w_G + x_V w_V + x_D w_D$ is the average fitness of the population.

Alternatively, we can write down these dynamics in their factored form (see appendix~\ref{app:factor} for a proof of equivalence) as:

\begin{align}
\dot{p} & = p(1 - p)(w_G - \langle w \rangle_{V,D}) \label{eq:pGen}\\
\dot{q} & = q(1 - q)(w_V - w_D) \label{eq:qGen}
\end{align}

\noindent where $p = x_G$ is the proportion of \GLY, $q = \frac{x_V}{x_V + x_D}$ is the proportion of aerobic cells that over-produce VEFG, and $\langle w \rangle_{V,D} = q w_V + (1 - q)w_D$ is the average fitness of the aerobic cells.

These equations are accurate for large populations at carrying capacity -- {\it in vivo} tumours up against a resource limitation or managed by an immune response -- or in their exponential growth phase -- typical of \emph{in vitro} experiments. 
In other cases they serve as an approximation.
See appendix~\ref{app:repDyn} for more discussion on interpreting replicator dynamics.

\section{Results}
\label{sec:results}

\begin{figure}
\centering
\begin{minipage}[t]{0.66\textwidth}\mbox{}\\[-\baselineskip]
\includegraphics[width=\textwidth]{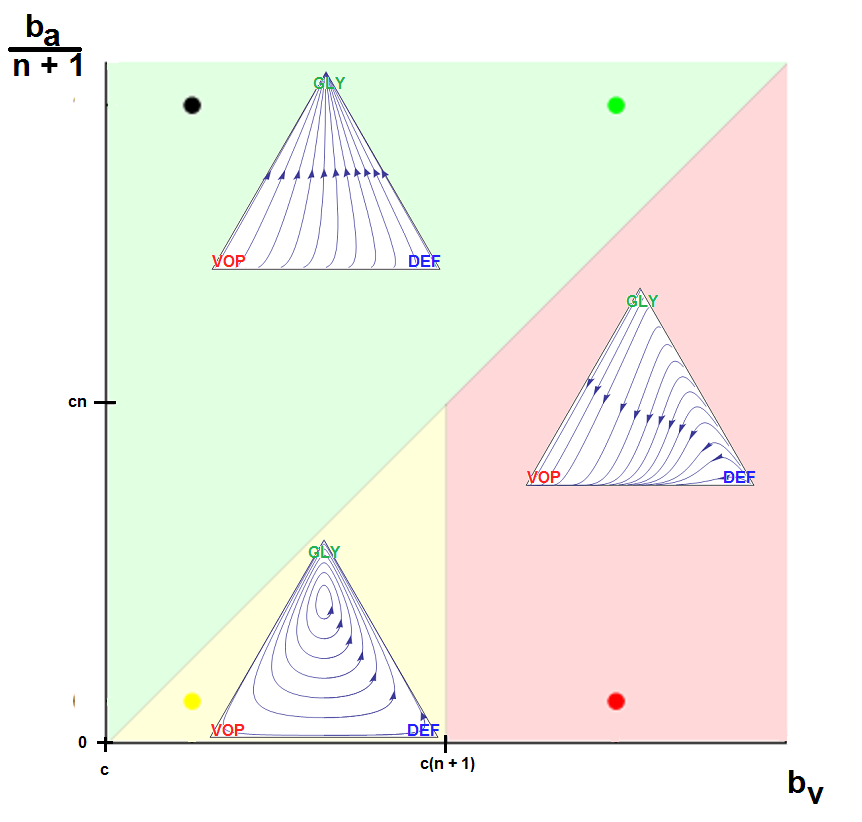}
\end{minipage}\hfill
\begin{minipage}[t]{0.33\textwidth}\mbox{}\\[-\baselineskip]
\caption{The $3$ possible dynamic regimes for the double goods game.
The possible parameter settings for $b_v$ are varied horizontally, starting at $c$.
The possible parameter settings for $\frac{b_a}{n + 1}$ are varied vertically, starting at $0$.
Each of the three inset simplexes have the same coordinates, with the top vertex corresponding to all \GLY, left to all \VOP, and right to all \DEF.
Each simplex is a typical example of dynamics within its regime. 
The specific micro-environmental parameters for each example: (1) $b_a = 37.5$, $b_v = 2$, $c = 1$, $n =4$ for green, (2) $b_a = 2.5$, $b_v = 7$, $c = 1$, $n = 4$ for red, and (3) $b_a = 2.5$, $b_v = 2$, $c = 1$, $n = 4$ for yellow.}
\end{minipage}
\label{fig:dynReg}
\end{figure}

We can apply the model described in the previous section to study the temporal evolution of different populations in various scenarios.
These scenarios are set by the four micro-environmental parameters of our model: $b_a$, the benefit per unit of acidification; $b_v$, the benefit from oxygen per unit of vascularization; $c$, the cost of (over)producing VEGF; and, $n$, the number of interaction partners in the public good.
Within a scenario, the last piece of information is the initial proportions of cells $x_{G}(0), x_{V}(0), x_{D}(0)$ (or $p(0), q(0)$ in the factored form).
The values of these variables will depend on the particular cancer ecosystem (i.e. patient and kind of tumour).
Since these variables can be difficult or impossible to measure clinically, it is important to understand what broad qualitative relationships between them mean for long-term dynamics.

Although the fitness functions in equations~\ref{eq:wG},~\ref{eq:wV}, and~\ref{eq:wD} are Bernstein polynomials of degree $n$, we can use properties of binomial coefficients to simplify the corresponding gain functions of the factored replicator dynamics without any approximation (see appendix~\ref{app:gain}).
This allows us to rewrite equations~\ref{eq:pGen}~and~\ref{eq:qGen} as:

\begin{align}
\dot{p} &= p(1 - p)(\overbrace{\frac{b_a}{n + 1} - q(b_v - c)}^{\text{gain function for } p}) \label{eq:pSimp}\\
\dot{q} &= q(1 - q)(\underbrace{\frac{b_v}{n + 1}[\sum_{k = 0}^n p^k] - c}_{\text{gain function for } q}) \label{eq:qSimp}
\end{align}

Notice that when $p = 0$, equation~\ref{eq:pSimp} recovers the social dilemma of angiogenesis that we discussed in the introduction with the free-riding \DEF\; cells taking over the population of aerobic cells, driving $q$ towards $0$.
When $p \neq 0$ and $q = 0$, we have a hypoxic tumour, and glycolytic cells are favored, driving $p$ towards $1$.
On the other hand, if $p \neq 0$ and $q = 1$, then we recover the social dilemma of aerobic glycolysis, with \VOP\; cells favoured if $b_a < (b_v - c)(n + 1)$, and \GLY\; cells otherwise.
In other words, if one of the strategies is absent in the population then no persistence of heterogeneity is possible among the remaining two strategies.
This is consistent with predictions from two-strategy linear public goods games.

The more interesting case that is unique to our model is when all three strategies are initially expressed in the population.
In this setting, we can analytically characterize the population dynamics into one of three qualitatively different regimes based on the values of the four main micro-environmental parameters by solving how the above functions for $p$ and $q$ cross zero (for more information, see appendix~\ref{app:gain}). 
This classification is shown visually in figure~\ref{fig:dynReg}.
We name these three regimes by their end points (and evolutionary stable strategies); fully glycolytic (green region in figure~\ref{fig:dynReg}), fully angiogenic (red), and heterogeneous (yellow).
The following three subsections describe each of these three regimes in turn.

\subsection{Fully glycolytic tumours: $\frac{b_a}{n + 1} > b_v - c$}

We solve for the average fitness of aerobic cells as:

\begin{equation}
\langle w \rangle_{V,D} = b_a p + q(b_v - c)
\label{eq:wVD}
\end{equation}

\noindent where the first summand is the benefit from acidification, and the second summand is the benefit due to oxygen from vascularization.
Since \DEF\; only consume the club good from vascularization, without producing any, equation~\ref{eq:wVD} is maximized to $b_a p + b_v - c$ when all aerobic cells are producers of the club good (i.e. all \VOP, q = 1). 

Similarily, we can solve for the fitness of \GLY:

\begin{equation}
w_G = b_a p + \frac{b_a}{n + 1}
\label{eq:wGsimp}
\end{equation}

\noindent where the first summand is the benefit from aciditification that all cancer cells receive, and the second summand is the slight increase in acidification that glycolytic cells get from always being in a group with an extra acid producer (themselves). 

From this, we see that if the fitness benefit of a single unit of acidification ($\frac{b}{n + 1}$) is higher than the maximum benefit from the club good for aerobic cells ($b_v - c$) then the difference between eq.~\ref{eq:wGsimp} and eq.~\ref{eq:wVD} is always positive. 
Thus, \GLY\; will always have a strictly higher fitness than aerobic cells, and be selected for. 
In this scenario, the population will converge towards all \GLY, regardless of the initial proportions (as long as there is at least some \GLY\;in the population).
This dynamic regime is achieved for any micro-environmental parameter settings corresponding to the green region in figure~\ref{fig:dynReg}.
Typical dynamics are shown in the inset simplex with micro-environmental parameters $b_a = 37.5, b_v = 2, c = 1, n = 4$.

\subsection{Fully angiogenic tumours: $\{ \frac{b_a}{n + 1}, cn \} < b_v - c$}

Consider an interaction group with $n_V$ \VOP\; and $n_D$ \DEF\;cells.
If the focal agent interacting with this group is an (over)producer then it will receive a benefit from oxygen of $\frac{b_v n_V + 1}{n_V + n_D + 1} - c$.
If the focal agent is a defector then they will receive $\frac{b_v n_V}{n_V + n_D + 1}$.
Since, by definition, $n_V + n_D$ is less than the interaction group size $n$ then regardless of the number of glycolytic cells:

\begin{equation}
w_V - w_D \geq \frac{b_v}{n + 1} - c
\label{eq:VDlow}
\end{equation}

Thus, if the benefit to (over)producers from their extra unit of vascularization ($\frac{b_v}{n + 1}$) is higher than the cost $c$ to produce that unit (or, equivalently, if $b_v - c > cn$) then \VOP\; will always have a strictly higher fitness than \DEF\;, selecting $q$ towards $1$.
In addition, if the maximum possible benefit of the club good to aerobic cells ($b_v - c$) is higher than the benefit of an extra unit of acidification ($\frac{b_a}{n + 1}$) then (based on the difference of eq.~\ref{eq:wGsimp} and~\ref{eq:wVD}, or the negation of the conditions in the previous section) for sufficiently high number of (over)producers ($q$ close enough to $1$), \GLY\; will have lower fitness than aerobic cells.
When both conditions are satisfied, the population will converge towards all \VOP.
This dynamic regime is achieved for any micro-environmental parameter settings corresponding to the red region in figure~\ref{fig:dynReg}.
Typical example dynamics are shown in the inset simplex with micro-environmental parameters $b_a = 2.5, b_v = 7, c = 1, n = 4$.
Notice how if a population starts with mostly aerobic cells ($p$ close to $0$) that are not overproducing VEGF ($q$ close to $0$) then on the way towards all-\VOP\;, the population might see a transient decrease in the number of aerobic cells.

\subsection{Heterogeneous tumours: $\frac{b_a}{n + 1} < b_v - c < cn$}
\label{sec:het}

From eq.~\ref{eq:VDlow}, we know that if the benefit from an extra unit of vascularization in a fully aerobic group ($\frac{b_v}{n + 1}$) is lower than the cost $c$ to produce that unit then for a sufficiently low proportion of \GLY\; and thus sufficiently large number of aerobic cells sharing the club good, \DEF\;will have higher fitness than \VOP. 
This will lead to a decrease in the proportion $q$ of (over)producers among aerobic cells and thus a decrease in the average fitness of aerobic cells $\langle w \rangle_{V,D}$ (see eq.~\ref{eq:wVD}). 
A lower fitness in aerobic cells will lead to an increase in the proportion of \GLY\; until the aerobic groups among which the club good is split get sufficiently small and fitness starts to favour \VOP\;over \DEF, swinging the dynamics back.
Thus, resulting in a rock-paper-scissor like dynamics.

Under this scenario, the population will orbit around an internal fixed-point at $q^* = \frac{b_a}{(b_v - c)(n + 1)}$ and $p^* \in (1 - \frac{b_v}{c(n + 1)}, \frac{c(n + 1)}{b_v} - 1)$. 
The exact position of $p^*$ is the solution to the polynomial equation $\sum_{k = 0}^n p^k = \frac{c(n + 1)}{b_v}$ (see appendix~\ref{app:gain} for details).
The amplitude of the orbit will depend on the distance between $p(0), q(0)$ and $p^*, q^*$.
This dynamic regime is achieved for any micro-environmental parameter settings corresponding to the yellow region in figure~\ref{fig:dynReg}.
Typical example dynamics are shown in the inset simplex with micro-environmental parameters $b_a = 2.5, b_v = 2, c = 1, n = 4$.

\section{Therapeutic Implications}

With the possible dynamic regimes in mind, we can think about treatment in one of two ways: (1) treatments that target the player by directly reducing the proportion of a given strategy in the population, or (2) treatments that target the game by changing the parameters ($b_a$, $b_v$, $c$ or $n$) and taking us from one dynamic regime to another.
In both cases we need to be mindful of counter-intuitive phenomena like timing and order effects and the importance of managing heterogeneity. 

\subsection{Treat the player: targeting cell-types}
\label{sec:treat_player}
\begin{figure}
\centering
\includegraphics[scale=0.80]{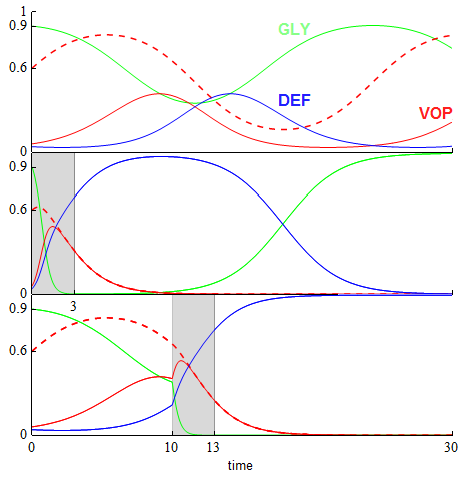}
\caption{Dynamics for an untreated tumour and examples of two different sets of dynamics resulting from changed timings of a given intervention.
Each graph is proportion of cells versus time, with \GLY\; ($x_G$) in green, \VOP\; ($x_V$) in solid red, and \DEF\; ($x_D$) in blue.
The dashed red lines show the proportion of \VOP\; among aerobic cells ($q = \frac{x_V}{x_V + x_D}$).
All three graphs start with same initial conditions ($p(0) (= x_G(0)) = 0.9$, $q(0) = 0.6$) and the same tumour micro-environment ($b_a = 2.5, b_v = 2, c = 1, n = 4$). 
In the second and third panel, we consider an anti-glycolytic treatment of the same strength (subtracting $3$ from the fitness of \GLY) and the same duration ($3$ time steps).
The difference between the second and third panel is in timing, marked in gray: the second panel starts treatment at $t =  0$, and the third at $t = 10$.
If the proportion of a cell-type goes below $10^{-4}$ then it dies off entirely.
At $t = 30$, no cell-types are extinct in the top panel; in the second and third panel, the only non-extinct cell-types are \GLY\; and \DEF, respectively.
This highlights the importance of the timing of therapy for evolutionary outcome.}
\label{fig:time}
\end{figure}

A treatment that targets a given strategy and can be applied long enough to drive that strategy to extinction can be considered a viable intervention.
If we were in the first dynamic regime (green in fig.~\ref{fig:dynReg}; $\frac{b_a}{n + 1} > b_v - c$) or the second dynamic regime (red in fig.~\ref{fig:dynReg}; $\{ \frac{b_a}{n + 1}, cn \} < b_v - c$) then the population will always converge towards all-\GLY\; or all-\VOP, respectively. 
This means that unless the strategy-targeting therapy is strong and long enough to drive that strategy to extinction, it will not affect the overall outcome beyond a potential transient delay.
In particular, the timing of the therapy will not have a qualitatively significant effect.
In the heterogeneous case (yellow in fig.~\ref{fig:dynReg}; $\frac{b_a}{n + 1} < b_v - c < cn$), however, counter-intuitive results are possible and the timing of treatment becomes important.

As an example in figure~\ref{fig:time}, we consider a tumour described by the micro-environmental parameters $b_a = 2.5, b_v = 2, c = 1$ and $n = 4$ and an initial composition $x_G = 0.9, x_V = 0.06, x_D = 0.04$.
If left untreated then the proportions of strategies would cycle around an internal fixed point at $(x_G^*,x_V^*,x_D^*) \approx (0.5,0.32,0.18)$, as seen in the top panel.
In that panel, the proportion of glycolytic cells $x_G$ will oscillate between about $0.9$ and $0.34$ with a period of about $26$ time units.

In the standard paradigm of personalized medicine, the most likely target would be the cell subtype composing the largest proportion of the tumour. In this case, it would be the glycolytic cells -- the most common, and least fit, strategy at $t = 0$. 
We can imagine targeting these cells specifically, with a therapy that imposes a large fitness cost, like lonidamine~\cite{Lonidamine}, for example. 
Here, we choose to set the fitness cost of this therapy to $3$, which leads to very quick and aggressive reduction in the \GLY\; population.
In a perfect world, the therapy would be applied long enough to drive \GLY\; to extinction.
But what happens if it is only applied for $3$ time units, enough to drastically reduce the proportion of \GLY\; cells -- below detecatble levels -- but not below the extinct threshold?
The cyclic dynamics then allow the glycolytic cells to recover by out-competing the mostly \DEF\; population.
To make matters worse, as \GLY\;  recovers, \VOP\; is pushed below the extinction threshold leaving just the other two-cells types to compete.
Without VEGF over-production, the remaining aerobic cells are less fit than the glycolytic cells and are also driven to extinction.
The overall result at $t = 30$ is a relapse with even more glycolytic cells than before treatment.
In this scenario, while the goal was to eliminate glycolytic cells, the opposite occurred: elimination of all the aerobic cells and creation of a completely glycolytic tumour.

If instead the treatment was delayed until $t = 10$ -- when the \GLY\;cells are at their lowest proportions, and highest fitness, in the tumour's natural cycle --  then much more favourable results could be achieved.
With the lower initial proportion of \GLY\; cells, $3$ units of time is long enough to drive the cells to extinction.
Without glycolytic cells, the competition between \DEF\; and \VOP\; becomes a classic social dilemma and the VEGF-overproducers are driven to extinction.
The result is now what was desired: an aerobic tumour with no -- or significantly diminished -- ability to recruit blood vessels.

In this case, we can think of the natural tumour dynamics as a neoadjuvant `treatment' that lowered the \GLY\; population slightly while preparing the \VOP-\DEF\;composition to be in a favourable position after therapy.

\subsection{Treat the game: targeting micro-environment}
\label{sec:treat_game}

\begin{figure}
\centering
\includegraphics[scale=0.80]{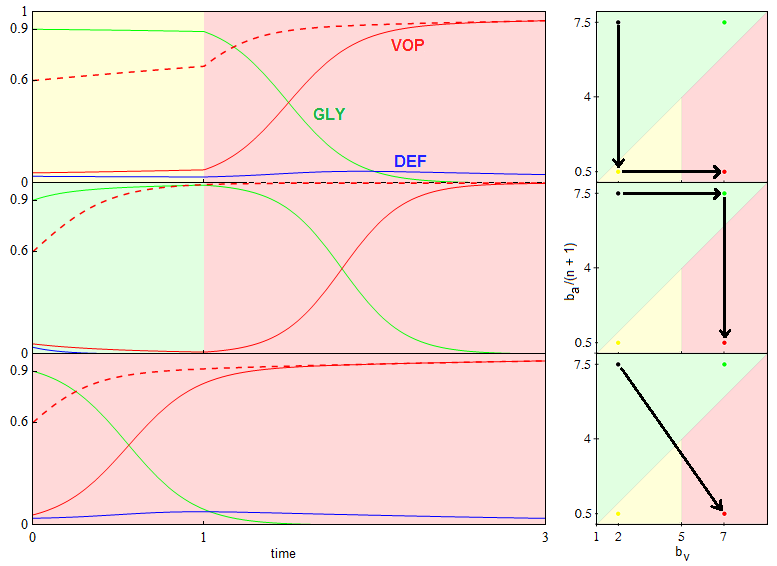}
\caption{Three possible orders of therapeutic intervention. 
Each graph on the left shows the proportion of cells versus time, with \GLY\; ($x_G$) in green, \VOP\; ($x_V$) in red, and \DEF\; ($x_D$) in blue.
The dashed red lines show the proportion of \VOP\; among aerobic cells ($q = \frac{x_V}{x_V + x_D}$).
Each graph on the right is $b_a/(n + 1)$ versus $b_v$ and shows how treatment moves the tumour between qualitatively different dynamic regimes through the space on micro-environmental parameters (for detailed explanation see figure~\ref{fig:dynReg}).
All graphs start with same initial proportions ($x_G(0) = 0.9, q(0) = 0.6$). The untreated tumour has parameters $b_a = 37.5$, $b_v = 2$, $c = 1$ and $n = 4$.
In the top two panels, two treatments are applied: the first at time $0$ and the second at time $1$. 
In the last panel, a single treatment is applied at time $0$.
We consider two game-targeting treatments: (1) a buffer therapy that reduces the benefit of acidity (setting $b_a = 2.5$ from then on; horizontal arrows in the right panel), and (2) a vascular renormalization therapy (VNT; setting $b_v = 7$ from then on; vertical arrows in the right panel).
In the top panel, buffer therapy is followed by VNT, the second panel shows VNT followed by buffer therapy, and in the final panel both treatments are given simultaneously (diagonal arrow in the right panel).}
\label{fig:order}
\end{figure}

Consider the hypothetical case-study in figure~\ref{fig:order}.
Here, we have a situation with poor vasculature ($b_v = 2$, $c = 1$, $n = 4$) and a highly glycolytic tumour ($p_0 = 0.9$, $b_a = 37.5$).
If left untreated, the tumour would quickly reach all-\GLY, driving the other two strategies extinct. 
Thus, the goal is to eliminate \GLY, and create an easier-to-target all \VOP\; tumour without incurring too much tumour heterogeneity.
We consider two possible interventions.
The treatments can be applied sequentially with a brief recovery window of $1$ time-step between them or simultaneously.
One treatment is a buffer therapy~\cite{bufferTherapy} that reduces the benefit of acidity; setting $b_a = 2.5$ from then on. 
Another is vascular normalization therapy (VNT)\cite{VNT}; setting $b_v = 7$ from then on. 
In figure~\ref{fig:order}, the top panel considers giving buffer therapy followed by VNT, the middle panel has the VNT preceding the buffer therapy, and the bottom panel has both treatments at given simultaneously at $t = 0$.

The buffer then VNT (top panel) ordering produces immediate results, with the proportion of \GLY\; no longer increasing -- and even starting to decline -- right away and reaching minimal levels earlier than VNT followed by buffer therapy (middle panel).
However, the top ordering increases the heterogeneity among the anaerobic cells, and although the tumour will eventually move to a state of all \VOP, by time-step $3$ (when it is nearly all \VOP\;for the middle panel) there is a high level of \VOP-\DEF\;heterogeneity, and earlier during treatment (say $t = 2$) the top panel has heterogeneity in all three cell types.
Something that the physician wanted to avoid.

In contrast, the middle ordering sees no immediate results from the VNT.
Instead, this first treatment can be thought of as a neoadjuvant therapy that eliminates the \VOP-\DEF\; heterogeneity among the rare aerobic cells before targeting the predominantly \GLY\; population.
By time-step $3$, the middle ordering sees a similarly high level of response in the \GLY\; phenotype, but without creating a high level of heterogeneity in the tumour.

We might expect that both normalizing the benefit due to oxygen from vascularization and decreasing the benefit from acidification at the same time would offer superior results to sequential therapy. 
However, as we can see from the bottom panel of figure~\ref{fig:order}, this is not necessarily the case. 
Although the same reduction in glycolytic cells is reached $1$ time unit faster than the sequential therapies, the heterogeneity among aerobic cells remains high; like buffer therapy followed by VNT.

Staggering buffer therapy after VNT might not produce immediately evident results but it lets us reach those results without encountering a highly heterogeneous tumour. When targeting the micro-environment, it matters which qualitatively different dynamic regimes the game goes through, even if the final micro-environmental parameters are the same.

\section{Discussion}

Both acidity and vascularization are goods that are costly to produce. 
But these goods are not exclusive to the producing cell and thus can benefit non-producers.
This poses an evolutionary social dilemma.
Non-producing cancer cells can free-ride off the producers by benefiting from these goods (for example, due to proximity to producers) while not paying the energetic costs of production.
Due to their higher relative fitness, such free-riders could then out-compete the producers, driving them to extinction (at a potential cost in absolute fitness to the tumour -- the society of cancer cells).
This raises the social dilemma of cooperation: if non-producers can free-ride then why do we see persistent heterogeneity in both aerobic metabolism~\cite{glyHetBreast1,glyHetBreast2,glyHetLung} and growth factor production~\cite{VEGF_Het,GenHet}? 

Prior work has considered in isolation the production of growth factors like VEGF~\cite{A13} and of acidity from glycolysis~\cite{A14}. 
They concluded that a heterogeneous equilibrium of producers and non-producers cannot exist unless the benefits that these goods provide are non-linear in the number of producers.
In the linear cases, considered separately, this would predict that VEGF (over)-producers (\VOP) and glycolytic (\GLY) cells would both go extinct, leaving a population of aerobic cells that do not call for more vasculature (\DEF).
This is in accord with the intuition that free-riders always win in social dilemmas, and would result in the elimination of heterogeneity of production of VEGF and acidity.

Instead of the separate analyses above, we recognized the natural coupling between the vasculature called by growth factor production and the acidity from glycolytic metabolism.
In particular, although the acidity from glycolysis is a good which is public to all tumour cells, the benefit from oxygen due to vasculature cannot be enjoyed by the glycolytic cells (\GLY).
Vascularization is a club good available only to the cells that continue to rely on aerobic metabolism (\VOP\;and \DEF).
In other words, a producer of the acidity good cannot be a beneficiary of the vasculature good.
We call such goods anti-correlated.

We show that the reductionist intuition leading to all-\DEF\;does not hold in this system of two anti-correlated linear goods.
Instead, we identify three qualitatively different dynamic regimes that end in a tumour that is either: (1) fully glycolytic (all-\GLY), (2) fully angiogenic (all-\VOP), or (3) heterogeneous (polyclonal) in all three cell types.
Which of these regimes is achieved depends on the micro-environmental parameters like the benefit per producer due to acidity ($b_a$), due to vasculature ($b_v$), and the cost of (over)producing VEGF ($c$).
Further, the heterogeneity is not static but maintained by a cycle in proportions of cellular strategies.
Thus, we show that polyclonal tumours made up of three different cellular strategies are evolutionary stable even with linear goods.
This stands in stark contrast to the all-\DEF\;equilibrium that we would expect from considering the linear goods in isolation.
Our results highlight the difficultly of ruling out possible dynamics from overly reductionist accounts of cancer, and the importance of modeling both the vascularization and acidity when studying the Warburg effect~\cite{Warburg2,Warburg3}.

The dynamic nature of the polyclonal equilibrium reminds us of the importance of tracking the tumour composition through time, not basing treatment on measurements from a single time-point, and optimizing the timing of treatment.
As an example, we consider an anti-glycolysis treatment like \emph{lonidamine}~\cite{Lonidamine}.
If timed correctly and applied for long enough, then this treatment can drive the glycolytic cells extinct and reduce the tumour to the two-strategy case of VEGF-production considered in prior work~\cite{A13}.
From there, somatic evolution will drive VEGF over-producers extinct, leaving us with an all-\DEF\;tumour (lower panel of figure~\ref{fig:time}).
However, if the same treatment is applied at the wrong time in the cycle of heterogeneity (or not for long enough) then the glycolytic population can recover while the VEGF over-producers are driven extinct by non-producers.
Without VEGF over-producers, the glycolytic cells can out-compete the aerobic cells and drive them to extinction, resulting in a fully glycolytic tumour (middle panel of figure~\ref{fig:time}).
A backfire effect for treatment.

Since the heterogeneous equilibrium is not the only possible outcome of these game dynamics, it is also important to measure the micro-environmental parameters like the benefit per unit of vascularization and per unit of aciditification that determine the game.
This dichotomy between tumour composition and micro-environmental parameters carries over from measurement to treatment.
In section~\ref{sec:treat_game}, we consider treatments like buffer therapy~\cite{bufferTherapy} and vascular renormalization therapy (VNT)\cite{VNT} to change the microenvironmental parameters and thus target the game.
By shifting to a more desirable game, we allow natural somatic evolution to lead us to a better outcome.
The order in which we shift between games is important, especially for transient heterogeneity.

For highly glycolytic tumours, it is important to consider neoadjuvant VNT prior to buffer therapy.
VNT allows us to reduce the heterogeneity in aerobic cells before targeting the more common glycolytic cells.
Thus, when buffer therapy turns the game against glycolytic cells, the tumour is prepped in a low heterogeneity state and moves to an all-\VOP\; phenotype without high levels of transient polyclonality.
If buffer therapy is applied before VNT, or even if the two are applied simultaneously, then the response of glycolytic cells is quicker but also prone to creating much longer lasting heterogeneity in the tumour.
We expect that similar considerations of neo-adjuvant therapy for managing the rarer cancer sub-types might prove effective in other cancer settings.

There are several different types of heterogeneity that can work for or against the patient.
These include strategy heterogeneity (i.e. a polyclonal tumour) and game heterogeneity (i.e. differences in micro-environmental parameters).
Physicians have to be mindful of both types when treating a given patient.
So far, we concentrated on difference in micro-environmental parameters due to variation between patients, tissues, and the effects of therapy.
However, there can also be game heterogeneity within different regions of the same tumour.
For example, we have previously shown~\cite{edge} that the game within the bulk of the tumour can differ from the game at static boundaries.
But it is through static boundaries like blood vessels, organ capsules, or basement membranes that metastatic invasion happens.  
If physicians want to minimize the risk of metastates then it is important to measure the microenvironmental parameters at such boundaries.
Future work could more explicitly model invasiveness by coupling our model to the go-vs-grow game~\cite{glyEGT1} and also broadening the analysis to \emph{ex situ} cells.

\putbib
\end{bibunit}
\newpage
\begin{bibunit} 

\begin{appendices}
\appendixpage
In these appendices, we develop and discuss the tools used to define and solve our model from section~\ref{sec:game}. The structure is as follows:
\begin{description}
\item{\ref{app:genModel}} Our model in general terms, with more detail than was possible in section~\ref{sec:game}.
We outline the linear benefit functions used in this paper, and possible generalizations to non-linear benefit functions for future work.

\item{\ref{app:factor}} Proof of a theorem for factoring any replicator dynamics and application of it to the specific case of factoring from $\{x_G,x_V,x_D\} \mapsto \{p,q\}$ (transforming eqs.~\ref{eq:xGgen},~\ref{eq:xVgen}~and~\ref{eq:xDgen} to eqs.~\ref{eq:pGen}~and~\ref{eq:qGen}) presented in this paper.

\item{\ref{app:repDyn}} Discuss of the generality of replicator dynamics with a focus on populations of constant size (\ref{app:constant}), exponentially growing populations (\ref{app:exponential}), and \emph{in vitro} systems (\ref{app:experimental}).

\item{\ref{app:gain}} Solution to the gain functions needed to transform the system in equations~\ref{eq:pGen}~and~\ref{eq:qGen} into equations~\ref{eq:pSimp}~and~\ref{eq:qSimp}.
We use this second description of our model to characterize the possible dynamics into the three qualitatively different regimes used in section~\ref{sec:results}.

\item{\ref{app:orbits}} Proof of closed orbits for the heterogeneous regime of section~\ref{sec:het}.
To achieve this, we show that -- after dynamic rescaling of time -- our system is Hamiltonian.
\end{description}

\section{General model for the double goods game}
\label{app:genModel}

Following Archetti~\cite{A14}, we consider the acidity produced by glycolysis, which affects normal cells more than cancer cells, as a public good for cancer cells. 
Acidity benefits all tumour cells, regardless of whether they have aerobic or anaerobic metabolism. 
In contrast, the increased oxygen from vascularization benefits only aerobic cells and so within a tumour it is a good that is private to non-glycolitic cells. 
Among just the aerobic tumour cells, following Archetti~\cite{A13}, vascularization is a public good since it benefits all aerobic cells in a local population regardless of how much VEGF -- or other vascularization increasing factors -- a particular individual is producing. 
Therefore, vascularization is a club good~\cite{clubgood}, and in game theoretic terms there are two kinds of defection. 
All tumour cells benefit from acidity but can defect from contributing to this public good by using aerobic metabolism. 
Any aerobic cell can benefit from higher vascularization but can defect from contributing to this club good by not producing, or producing less, VEGF. 
Since being a defector in the public good is required to benefit from the club good, we call the two goods anti-correlated.

If we consider a cell interacting with $n$ other nearby cells and define $A_n(k) = \frac{b_a k}{n}$ as the benefit due to acidity if $k$ cells are producing acid, and $V_n(k) = \frac{b_v k}{n}$ as the benefit (to aerobic cells) due to vascularization if $k$ cells are (over) producing VEGF at a cost $c$. 
This results in three pure strategies with the following payoffs:

\begin{description}
\item{\GLY} The glycolytic strategy does not use oxygen and thus produces acid.
It does not (over) produce VEGF to call for more blood vessels and thus does not increase vascularization. When interacting with $n_G$ other \GLY\;cells, its payoff is: $A_n(n_G + 1)$.
\item{\VOP} This strategy uses oxygen and thus does not produce acid. It also (over) produces VEGF at a cost $c$ and thus calls for more blood vessels, directly increasing vascularization. 
When interacting with $n_G$ \GLY\;cells and $n_V$ other \VOP\;cells, its payoff is: $A_n(n_G) + V_{n - n_G}(n_V + 1) - c$. 
Note that the oxygen benefit due to increased vasculization is shared only between non-glycolitic cells, hence the club good group size for $V$ is $n - n_G$.
\item{\DEF} This strategy is the universal defector, it uses oxygen and thus does not produce acid, and it also does not (over) produce VEGF. 
When interacting with $n_G$ \GLY\; cells and $n_V$ \VOP\; cells, its payoff is: $A_n(n_G) + V_{n - n_G}(n_V)$.
\end{description}

We do not explicitly consider the case of a strategy that does not use oxygen but produces VEGF;
when the cost of (over) producing VEGF $c > 0$, the fitness of this fourth strategy will always be strictly lower than \GLY\;and thus always driven to extinction. Across our population of cells, our fitness functions are thus:

\begin{align}
w_G & = \sum_{n_G + n_V + n_D = n} \binom{n}{n_G,n_V,n_D} x_G^{n_G} x_V^{n_V} x_D^{n_D} A_n(n_G + 1) \\
& = \sum_{n_G = 0}^n \binom{n}{n_G} x_G^{n_G}(1 - x_G)^{n - n_G} A_n(n_G + 1) \\
& = \langle A_n(n_G + 1) \rangle_{n_G \sim \mathbf{B}_n(x_G)} \\
w_V & = \sum_{n_G + n_V + n_D = n} \binom{n}{n_G,n_V,n_D} x_G^{n_G} x_V^{n_V} x_D^{n_D} ( A_n(n_G) + V_{n - n_G}(n_V + 1) - c ) \\
& = \langle A_n(n_G) \rangle_{n_G \sim \mathbf{B}_n(x_G)} + \langle V_{n - n_G}(n_V + 1) \rangle_{(n_G,n_V) \sim \mathbf{M}_n(x_G,x_V)} - c \\
w_D & = \sum_{n_G + n_V + n_D = n} \binom{n}{n_G,n_V,n_D} x_G^{n_G} x_V^{n_V} x_D^{n_D} ( A_n(n_G) + V_{n - n_G}(n_V) ) \\
& = \langle A_n(n_G) \rangle_{n_G \sim \mathbf{B}_n(x_G)} + \langle V_{n - n_G}(n_V) \rangle_{(n_G,n_V) \sim \mathbf{M}_n(x_G,x_V)}
\end{align}

\noindent where $\mathbf{B}_n(x_G)$ is the binomial distribution with $n$ samples and $x_G$ is the probability of success; $\mathbf{M}_n(x_G,x_V)$ is the multinomial distribution with $n$ samples and $x_G$ is the probability of the first outcome, $x_V$ of the second; and $\langle f(k) \rangle_{k \sim \mathbf{K}}$ means the expected value of $f(k)$ over $k$ sampled from the distribution $\mathbf{K}$.

If we define the proportion of glycolytic cells as $p := x_G$ and the proportion of VEGF producers among the aerobic cells as $q := \frac{x_V}{x_V + x_D}$, then we can express the replicator dynamics in the factored form (see appendix~\ref{app:factor}):

\begin{align}
\dot{p} &= p(1 - p)(w_G - \langle w \rangle_{V,D}) \\
\dot{q} &= q(1 - q)(w_V - w_D)
\end{align}

\noindent where $\langle w \rangle_{V,D} = q w_V + (1 - q)w_D$. 

We discuss the generality and applicability of replicator dynamics in appendix~\ref{app:repDyn}. In appendix~\ref{app:gain}, we show how the gain functions of these dynamics are simplified significantly in the linear case that we consider in this paper.
Our focus on the linear case allows us to borrow many tools of analysis from the optional public-goods game.~\cite{VPG,ourVPG}
More complicated functional forms of $A_n$ and $V_n$ could be considered, and yield a fascinating family of possible dynamics, but this is beyond the scope of the current paper.
For an initial treatment of the non-linear case, see Kaznatcheev~\cite{nonlinDPG}.
In the case of $p = 0$ the non-linear case reduces to~\cite{A13}, and in the case of $q = 0$ or $q = 1$ the non-linear case reduces to~\cite{A14}.

Mathematica notebooks for visualizing and interacting with the model are available at \url{https://github.com/kaznatcheev/MathOnco} as \emph{dgLinB.nb}, \emph{dgLinB\_Simplexes.nb}, and \emph{dgLinB\_TreatPlayer.nb}.

\section{Factoring the replicator equation}
\label{app:factor}

In this section, we are going to present a general trick for factoring replicator dynamics into a nested form. 
We will then apply that trick to transform $\{x_G,x_V,x_D\} \mapsto \{p,q\}$. 
An early form of the account in this appendix first appeared in~\cite{factor}, and Hauert et al.~\cite{VPG} previously considered the special case of optional public goods game.

\begin{mydef}
A compound strategy $(Y_j,W_j,c_j)$ of a set of strategies $\{x_i,w_i\}_{i = 1}^n$ is defined by:
\begin{itemize}
\item a component function $c_j$ with $\sum_{i = 1}^n c_j(i) = 1$,
\item a weight $Y_j := \sum_i^n c_j(i)x_i$
\item a profile $y_{ij} = \frac{c_j(i)x_i}{Y_j}$, and
\item a fitness $W_j = \sum_{i = 1}^n  y_{ij} w_i$.
\end{itemize}
\end{mydef}

\begin{mydef}
A set of compound strategies $\{(Y_j,W_j,c_j)\}_{j = 1}^m$ factors a set of strategies $\{(x_i,w_i)\}_{i = 1}^n$ with $x_i = \sum_{j:c_j(i) \neq 0} \frac{y_{ij}Y_j}{c_j(i)}$ if for all $1 \leq i \leq n$, $\sum_{j = 1}^m c_j(i) = 1$.
\end{mydef}

\begin{thm}
If $\{(Y_j,W_j,c_j)\}_{j = 1}^m$ factors $\{(x_i,w_i)\}_{i = 1}^n$ then the system
\begin{equation}
\dot{x_i} = x_i(w_i - \langle w \rangle)
\end{equation}

where $\langle w \rangle = \sum_{k = 1}^n x_kw_k$, and the system

\begin{align}
\dot{Y_j} & = Y_J(W_j - \langle W \rangle) \\
\dot{y_{ij}} & = y_{ij}(w_i - \langle w \rangle_j)
\end{align}

where $\langle W \rangle = \sum_{k = 1}^m Y_k W_k$ and $\langle w \rangle_j = \sum_{k = 1}^n y_{kj}w_k $describe the same dynamics.
\end{thm}

\begin{proof}
First, notice that $\langle w \rangle_j := \sum_i y_{ij}w_i = W_j$, and 
\begin{align}
\langle W \rangle & = \sum_j Y_j W_j  = \sum_j Y_j \sum_i y_{ij} w_i \\
& = \sum_{i,j} c_j(i) x_i w_i = \sum_i x_i w_i \\
& = \langle w \rangle.
\end{align}

Now check that the dynamics are equal:

\begin{align}
\dot{x_i} & = \frac{d}{dt}(\sum_{j:} c_j(i) x_i  = \frac{d}{dt}(\sum_{j:c_j(i) \neq 0} c_j(i) \frac{y_{ij}Y_j}{c_j(i)})  \\
& = \sum_j \frac{d}{dt}(y_{ij}Y_j)  = \sum_j (\dot{y_{ij}}Y_j + y_{ij}\dot{Y_j}) \\
& = (\sum_j y_{ij}(w_i - \langle w \rangle_j)Y_j + y_{ij}Y_{j}(W_{j} - \langle W \rangle)) \\
& = \sum_j y_{ij}Y_j(w_i - \langle w \rangle)  = (\sum_j c_j(i) x_i)(w_i - \langle w \rangle) \\ 
& = x_i(w_i - \langle w \rangle)
\end{align}

\end{proof}

We can apply the above theorem to the particular game in section~\ref{sec:game} and appendix~\ref{app:genModel}.

\begin{cor}
The dynamics given in section~\ref{sec:game} for \GLY-\VOP-\DEF\; competition:
\begin{align}
\dot{x_G} = x_G(w_G - \langle w \rangle) \\
\dot{x_V} = x_V(w_V - \langle w \rangle) \\
\dot{x_D} = x_D(w_D - \langle w \rangle)
\end{align}
are equivalent to their factored form for glycolysis and angiogenesis competition:
\begin{align}
\dot{p} = p(1 - p)(w_G - \langle w \rangle_{V,D}) \\
\dot{q} = q(1 - q)(w_V - w_D)
\end{align}
where $\langle w \rangle_{V,D} = qw_V + (1 - q)w_D$.
\end{cor}

\begin{proof}
Use the component functions $c_A(G) = 1, c_A(V) = 0, c_A(D) = 0$ and $c_B(G) = 0, c_B(V) = 1, c_B(D) = 1$ and name $p := Y_A = x_G$ and $q := y_{B,V} =  \frac{x_V}{x_V + x_D}$. Notice that $W_A = w_G$ and $W_B (=  \langle w \rangle_B) =\langle w \rangle_{V,D}$, giving us the equations:
\begin{align}
\dot{p} & = p(w_G - \langle w \rangle) \\
& = p(w_G - pw_G - (1 - p)\langle w \rangle_{V,D}) \\
& = p(1 - p)(w_G - \langle w \rangle_{V,D}) \\
\dot{q} & = q(w_V - \langle w \rangle_{V,D}) \\
& = q(w_V - qw_V - (1 - q)w_D) \\
& = q(1 - q)(w_V - w_D)
\end{align}
\end{proof}

\section{Replicator dynamics for constant and exponentially growing populations}
\label{app:repDyn}

In the late 1970s, replicator dynamics were introduced into theoretical biology and ecology literature by a number of authors~\cite{repDyn1,repDyn3,repDyn4,MS82}.
As we discussed in the introduction, they have been used extensively since then, coming into mathematical oncology near the turn of the millennium~\cite{EGTOnco1,EGTOnco2}.
However, confusion still exists about the range of applicability of replicator dynamics.
In this section, we will describe three different ways to get the replicator equation from micro-dynamical or experimental considerations:
\begin{description}
\item{\ref{app:constant}} We discuss Traulsen et al.'s~\cite{constPop} treatment of constant size.
\item{\ref{app:exponential}} We build up from an exponential growth model.
\item{\ref{app:experimental}} We will start from an experimental system with arbitrary \emph{in vitro} growth.
\end{description}
None of these micro-dynamical foundations perfectly capture the \emph{in vivo} conditions of a cancer patient.
But the robustness of the replicator equations to changes in founding assumptions should convince us that they are a good starting point for the heuristic model in this paper. 
This appendix closely follows Kaznatcheev~\cite{manyRepDyn}.

\subsection{Fitness as probability to reproduce in populations of constant size}
\label{app:constant}

A focus on individual agents pushes the theorist to look for definitions of fitness on the individual -- rather than population -- level. 
We find such a definition in the Moran process~\cite{moran1,moran2}, and with it comes the realization of replicator dynamics as a process on well-mixed and large but fixed sized populations.

In a Moran process, we imagine that a population is made up of a fixed number $N$ of individuals. 
An agent is selected to reproduce in proportion to their game payoff, and their offspring replaces another agent in the population, chosen uniformly at random. 
This gives us a very clear individual account of fitness as a measure of the probability to place a replicate into the population.

Traulsen, Claussen, and Hauert~\cite{constPop} wrote down the Fokker-Planck equation for the above Moran process, and then use Ito-calculus to derive a Langevin equation for the evolution of the proportions of each strategy $x_k$. The fluctuations in this stochastic equation scale with $1/\sqrt{N}$ and so vanish in the limit of large $N$. This reduces them to a deterministic limit of the replicator equation in Maynard Smith form~\cite{MS82}, with the fitness functions as the payoff functions:

\begin{equation}
\dot{x}_k = x_k\frac{w_k - \langle w \rangle}{\langle w \rangle}
\end{equation}

\noindent where the extract condition of $\langle w \rangle > 0$ is introduced. 
In the limit of weak-selection, this form is equivalent to the Taylor form~\cite{repDyn1} used in the body of the paper:

\begin{equation}
\dot{x}_k = x_k(w_k - \langle w \rangle)
\end{equation}

Even without weak-selection, the two systems of equations differ only by dynamic time rescaling and thus have the same fixed points, orbits, and paths. 
Since these are often the things we care most about during our analyses, we can use the equations interchangeably.

\subsection{Exponentially growing populations}
\label{app:exponential}

Suppose that for some theoretical reasons, we think that the microdynamics of our population is well represented by an exponential growth model. 
Consider $m$ types of cells with $N_1, ... , N_m$ individuals each, growing with growth rates $w_1, ..., w_m$, which could be functions of various other parameters. The population dynamics are then described by the set of $m$ differential equations: $\dot{N}_k = w_k N_k$ for $1 \leq k \leq m$. 

Let $N = N_1 + ... + N_m$, and look at the dynamics of $x_k = N_k/N$:

\begin{align}
\dot{x}_k & = \frac{\dot{N}_k}{N} - \frac{N_k\dot{N}}{N^2} \\
& = \frac{w_k N_k}{N} - \frac{N_k}{N}\frac{\sum_{i = 1}^m w_i N_i}{N} \\
& = \frac{N_k}{N}(w_k - \sum_{i = 1}^m w_i \frac{N_i}{N}) \\
& = x_k(w_k - \langle w \rangle)
\end{align}

which is just the replicator dynamics. 
Thus, replicator dynamics perfectly describe an exponentially growing population.

\subsection{Fitness as experimental growth rate}
\label{app:experimental}

Alternatively, we can start from measurement. 
In particular, measurements of growth rates as would be done for \emph{in vitro} experiments. 
We will recover replicator dynamics from purely experimental outputs or their limits.

If we are running an experiment with some tagged cells of $m$ many types: $1, ... , m$ then our most basic primitives are (estimates of) the sizes of the seeding populations $N^I_1, ... , N^I_m$ and the size of the final populations $N^F_1, ... , N^F_m$. 
Given these values and a time difference $\Delta t$ between when the experiment was started and when it ended, the experimental population growth rates are:

\begin{equation}
w_k := \frac{N^F_k - N^I_k}{N^I_k \Delta t}
\end{equation}

this can be rotated to give us a mapping $N_I \mapsto N_F$:

\begin{equation}
N^F_k = N^I_k (1 + w_k\Delta t)
\end{equation}

From defining the initial and final population sizes $N^{\{I,F\}} = \sum_{k=1}^m N^{\{I,F\}}_k$, we can compare the initial and final proportions of each cell type:

\begin{align}
x^I_k & = \frac{N^I_k}{N^I} \\
x^F_k & = \frac{N^F_k}{N^F} \\
& = x^I_k \frac{1 + w_k\Delta t}{1 + \langle w \rangle \Delta t}
\end{align}

where $\langle w \rangle = \sum_{k = 1}^m x^I_k w_k$.

So far we were looking at a discrete process. 
But we can approximate it with a continuous one. 
In that case, we can define $x_k(t) = x^I_k,\;\; x_k(t + \Delta t) = x^F_k$ and look at the limit as $\Delta t$ gets very small:

\begin{align}
\dot{x} & = \lim_{\Delta t \rightarrow 0} \frac{x_k(t + \Delta t) - x_k(t)}{\Delta t} \\
& = \lim_{\Delta t \rightarrow 0} \frac{x^I_k}{\Delta t} ( \frac{1 + w_k\Delta t}{1 + \langle w \rangle \Delta t} - 1) \\
& = \lim_{\Delta t \rightarrow 0} x_k \frac{w_k - \langle w \rangle}{1 + \langle w \rangle \Delta t} \\
& = x_k(w_k - \langle w \rangle)
\end{align}

\noindent Therefore, we recover replicator dynamics and give an explicit experimental interpretation for all of our theoretical terms.

\section{Gain functions to characterize dynamics}
\label{app:gain}

This section closely follows \cite{VPG,linDPG}.

In the linear public goods that we consider in this paper, each contributor increases the public (or club) good by a constant amount and the whole good is divided evenly among the whole public (or club) regardless of if they contributed. In symbols: $A_n(k) = \frac{b_a k}{n + 1}$ and $V_n(k) = \frac{b_v k}{n + 1}$. 

This gives us the gain functions for $q$:

\begin{align}
w_V - w_D & = \sum_{n_G + n_V + n_D = n} \binom{n}{n_G,n_V,n_D} x_G^{n_G} x_V^{n_V} x_D^{n_D} (\frac{b_v}{n - n_G + 1} - c) \\
& = \sum_{m = 0}^n \binom{n}{m} p^{n - m}(1 - p)^m \frac{b_v}{m + 1} - c
\end{align}

\noindent where we relabeled by defining $m := n_V + n_D (= n - n_G)$ as the number of participants in the club good. 
The issue now is to eliminate the $m + 1$ in the denominator which can be done by observing that $\binom{n}{m} = \frac{m + 1}{n + 1}\binom{n + 1}{m + 1}$.
Resuming:

\begin{align}
w_V - w_D & = \sum_{m = 0}^n \binom{n + 1}{m + 1} p^{n - m}(1 - p)^m \frac{b_v}{n + 1} - c \\
& = = \frac{b_v}{(1 - p)(n + 1)}([\sum_{m' = 0}^{n'}\binom{n'}{m'}p^{n' - m'}(1 - p)^{m'}] - \binom{n'}{0}p^{n'}) - c\\
& = \frac{b_v}{(1 - p)(n + 1)}(1 - p^{n + 1})- c, \label{eq:VDgain} \\
& = \frac{b_v}{n + 1}[\sum_{k = 0}^n p^k] - c
\end{align}

\noindent where in the first to second step, we relabel with $n' = n + 1$ and $m' = m + 1$ to get the binomial distribution. 
The final step follows from polynomial division. 

Continuing on to the edge between \GLY\; and \DEF:

\begin{align}
w_G - w_D & = \sum_{n_G + n_V + n_D = n} \binom{n}{n_G,n_V,n_D} x_G^{n_G} x_V^{n_V} x_D^{n_D} (\frac{b_a}{n + 1}  - \frac{b_v n_V}{n - n_G + 1}) \\
& = \frac{b_a}{n + 1} - \Big( \sum_{m = 0}^n \binom{n}{m} p^{n - m} (1 - p)^m \Big( \sum_{k = 0}^m \binom{m}{k} q^k(1 - q)^{m - k} \frac{b_v k }{m + 1} \Big) \Big) \\
& = \frac{b_a}{n + 1} - \sum_{m = 0}^n \binom{n}{m} p^{n - m} (1 - p)^m  b_v q (1 - \frac{1}{m + 1}) \\
& = \frac{b_a}{n + 1}  - b_v q + \frac{b_v q}{(1 - p)(n + 1)}(1 - p^{n + 1}). \label{eq:GDgain}
\end{align}

Which means that the gain function for $p$ is:

\begin{align}
w_G - \langle w \rangle_{V,D} & = w_G - (q w_V + (1 - q)w_D) \\
& = w_G - w_D - q(w_V - w_D) \\
& = \frac{b_a}{n + 1}  - b_v q + \frac{b_v q}{(1 - p)(n + 1)}(1 - p^{n + 1}) - q(\frac{b_v}{(1 - p)(n + 1)}(1 - p^{n + 1}) - c) \\
& = \frac{b_a}{n + 1} - q(b_v - c).\label{eq:Ggain}
\end{align}

Thus, our replicator dynamics are given by:

\begin{align}
\dot{p} &= p(1 - p)(\overbrace{\frac{b_a}{n + 1} - q(b_v - c))}^{\text{gain function for } p} \label{eq:repP}\\
\dot{q} &= q(1 - q)(\underbrace{\frac{b_v}{n + 1}[\sum_{k = 0}^n p^k] - c)}_{\text{gain function for } q} \label{eq:repQ}
\end{align}

To find the boundaries between the dynamics, we then have to solve for when the gain functions in eqs.~\ref{eq:VDgain} and~\ref{eq:Ggain} are equal to zero.
From these gain functions, we can see that the the linear goods can have three possible dynamics.
First, consider when the gain function for $p$ crosses $0$:

\begin{itemize}
\item If $b_a > (b_v - c)(n + 1)$ then the gain function for $p$ is positive regardless of $q$, and the population converges towards all \GLY\;(green in figure~\ref{fig:dynReg}), else
\item if $b_a < (b_v - c)(n + 1)$ then we have consider when the gain function for $q$ crosses $0$. 
One of two cases is possible:
\begin{itemize}
\item If $b_v > c(n + 1)$ then regardless of the value of $p$, the gain function for $q$ is positive, so $q \rightarrow 1$ making the gain function for $p$ negative and the population converges towards all \VOP\;(red in figure~\ref{fig:dynReg}), else
\item if $b_v < c(n + 1)$ then the population will orbit around an internal fixed-point at $q^* = \frac{b_a}{(b_v - c)(n + 1)}$ (i.e. the zero of the gain function for $p$) and $p^*$ as the unique positive root of $\sum_{k = 0}^n p^k = \frac{c(n + 1)}{b_v}$ (i.e. the zero of the gain function for $q$). 
Thus, $1 - \frac{b_v}{c(n + 1)} \leq p^* \leq \frac{c(n + 1)}{b_v} - 1$. 
This is the yellow region in figure~\ref{fig:dynReg}.
\end{itemize}
\end{itemize}

For a proof of closed orbits, see appendix~\ref{app:orbits}.

\section{Orbits in a Hamiltonian System}
\label{app:orbits}

Let us rigorously establish the existence of closed orbits around the internal fixed point in the third dynamic region (yellow in figure~\ref{fig:dynReg}). For this, we will follow closely the analysis of the optional public good by Hauert et al.~\cite{VPG} as specified to this model in~\cite{cycles}.

Consider the equations governing the dynamics of the linear goods as given in eqs.~\ref{eq:repP} and~\ref{eq:repQ}.
Since $pq(1 - p)(1 - q)$ is strictly positive for $p,q \in (0,1)$, dividing the right hand side of the system in eqs.~\ref{eq:repP} and~\ref{eq:repQ} by $pq(1 - p)(1 - q)$ results only in a change of rate (dynamic rescaling of time). 
Thus, the resulting system preserves the orbits of our replicator dynamics. 
This new system is given by:

\begin{align}
\dot{p} &= \frac{b_a}{(n + 1)q(1 - q)} - \frac{b_v - c}{1 - q} \label{eq:HamP}\\
\dot{q} &= \frac{b_v}{(n + 1)p(1 - p)}[\sum_{k = 0}^n p^k] - \frac{c}{p(1 - p)} \label{eq:HamQshort} \\
&= \frac{b_v}{n + 1}\Big( \frac{1}{p(1 - p)} + \frac{n}{1 - p} - \sum_{k = 0}^{n - 2}(n - 1 - k)p^k \Big) - \frac{c}{p(1 - p)} \label{eq:HamQlong}
\end{align}

where eq.~\ref{eq:HamQlong} follows from eq.~\ref{eq:HamQshort} by dividing $\sum_{k = 0}^n p^k$ by $p(1 - p)$.

Define the following primitives:

\begin{align}
R(p) & := \frac{b_v - c(n + 1)}{n + 1}\ln \frac{p}{1 - p} - \frac{b_v}{n + 1} ( n \ln (1 - p) + \sum_{k = 0}^{n - 2} \frac{n - 1 - k}{k + 1}p^{k + 1})\\
S(q) & := \frac{b_a}{n + 1}\ln\frac{q}{1 - q} + (b_v - c) \ln (1 - q)
\end{align}

and the Hamiltonian $H = R - S$. 

Notice that eq.~\ref{eq:HamP}~and~\ref{eq:HamQlong} can be rewritten as $\dot{p} = - \frac{\partial H}{\partial q}, \dot{q} = \frac{\partial H}{\partial p}$. 
Thus, we have a 1-dimensional Hamiltonian system. 
When our conditions $b_a < (b_v - c)(n + 1)$ and $b_v < c(n + 1)$ are met then $(p^*,q^*)$ from appendix~\ref{app:gain} is a unique strict global minimum of $H$ in $(0,1)^2$. 
Thus  $(p^*,q^*)$ is a center of the dynamics with orbits of constant $H$. 
Finally, as $p$ or $q$ go to $0$ or $1$, $H(p,q)$ goes to $\infty$ if our conditions are met; thus the orbits are internal to $(0,1)^2$ and closed. 
Since our transformations preserved orbits, this means that the replicator dynamics of linear public and club goods have the same closed orbits.

\putbib
\end{appendices}
\end{bibunit}

\end{document}